\documentclass[letterpaper, 10 pt, conference]{ieeeconf}  
\IEEEoverridecommandlockouts
\usepackage{balance}
\usepackage{color}
\usepackage{caption}
\usepackage{enumerate}
\usepackage{cite}
\usepackage{empheq}
\usepackage{mathrsfs}

\usepackage{mathtools}

\usepackage{tikz}

\usepackage[font=small,labelfont=bf]{caption}

\usepackage{flushend}

\newcommand{\tdiam}{\text{diam}}
\newcommand{\ttdiam}{\textstyle{\mathrm{diam}}}

\newcommand{\tdiag}{\text{diag}}
\newcommand{\NP}{\text{NP}}
\newcommand{\GP}{\text{GP}}

\newcommand{\K}{\mathbb{K}}

\newcommand{\real}{\ensuremath{\mathbb{R}}}

\newcommand{\until}[1]{\{1,\dots, #1\}}
\newcommand{\subscr}[2]{#1_{\textup{#2}}}
\newcommand{\supscr}[2]{#1^{\textup{#2}}}

\newcommand{\adj}{\operatorname{\textstyle{\mathrm{Adj}}}}
\usepackage{subfig}


\makeatletter

\makeatletter
\let\proof\@undefined                        
\let\endproof\@undefined                  
\makeatother
\usepackage{graphicx,amssymb,amstext,amsmath,amsthm}

\usepackage[bookmarks=true]{hyperref}
\usepackage{algorithm,algorithmicx,algpseudocode}
\algnewcommand{\algorithmicgoto}{\textbf{go to}}%
\algnewcommand{\Goto}[1]{\algorithmicgoto~\ref{#1}}%
\algnewcommand{\LineComment}[1]{\Statex \(\triangleright\) #1}
\algnewcommand{\LineCommentN}[1]{\Statex \hspace{1cm}\(\triangleright\) #1}

\usepackage{multirow}
\usepackage{stfloats}

\newcommand{\st}{\operatorname{s.t.}}

\newtheorem{prop}{Proposition} 

\newtheorem{thm}{Theorem}
	
\newtheorem{lem}{Lemma}
\newtheorem{defn}{Definition}
\newtheorem{rem}{Remark}
\newtheorem{problem}{Problem}

\usepackage{setspace}
\let\oldbibliography\thebibliography
\renewcommand{\thebibliography}[1]{%
  \oldbibliography{#1}%
}

\newcommand{\mk}[1]{{\color{black} #1}}

\newcommand{\md}[1]{{\color{black} #1}}

\newcommand{\mh}[1]{{\color{black} #1}}

\newcommand{\sm}[1]{{\color{black} #1}}

\newcommand{\moo}[1]{{\color{black} #1}}
\newcommand{\mkj}[1]{{\color{black} #1}}

\newcommand{\mmoh}[1]{{\color{black} #1}}
\newcommand{\mmok}[1]{{\color{black} #1}}
\begin{document}

\title{\LARGE \bf \md{Guaranteed Privacy-Preserving $\mathcal{H}_{\infty}$-Optimal Interval Observer \mkj{Design} for Bounded-Error LTI Systems}} 

\author{%
Mohammad Khajenejad and Sonia Mart{\'\i}nez\\
\thanks{
{M. Khajenejad and S. Mart{\'\i}nez are with the Department of
  Mechanical and Aerospace Engineering, University of California, San
  Diego, San Diego, CA, USA (e-mail: \{mkhajenejad,
  soniamd\}@ucsd.edu).} \mkj{This work is supported by the grants ONR N00014-19-1-2471 and NSF NSF 2003517,3.}} 
}

\maketitle
\thispagestyle{empty}
\pagestyle{empty}

\begin{abstract}
  This paper \sm{furthers current research into the notion of}
  \emph{guaranteed privacy}, which provides a deterministic
  characterization of the \sm{privacy of output signals of a dynamical
    system or mechanism}.  Unlike stochastic differential privacy,
  guaranteed privacy offers strict bounds on the proximity between the
  ranges of two sets of estimated data.  \sm{Our approach relies on}
  synthesizing an interval observer \mkj{for a perturbed linear time-invariant
    (LTI) bounded-error system. The design procedure} incorporates
    \mkj{a} bounded noise
  perturbation factor computation and observer gain\mmoh{s} synthesis. Consequently, \mkj{the} observer
  simultaneously provides guaranteed private and stable
  interval-valued estimates for \mmoh{a} desired variable. We demonstrate
  the optimality of our design by minimizing the
  $\mathcal{H}_{\infty}$ norm of the observer error system. Furthermore, we
  assess the accuracy of our proposed mechanism by quantifying the
  loss incurred when considering guaranteed privacy specifications. Finally, we illustrate the outperformance of the proposed approach to differential privacy through simulations.
  \end{abstract}
  \vspace{-0.1cm}
  
\section{Introduction}\label{sec:intro} 
The preservation of data privacy and security has become a pivotal
concern in the oversight of cyber-physical systems (CPS) and their
public credibility. Malicious actors can expand the scope of their
attacks by extracting valuable information from the numerous physical,
control, and communication components of the system; inflicting harm
upon both the CPS and its users. \sm{While this data may initially be
  hidden, such information may be inferred by the examination of other
  mixed data, which made available either unintentionally or to
  provide a system-wide service.} Consequently, a significant
endeavor is underway to develop resilient control strategies that
ensure data security within these systems
\cite{cortes2016differential}. \sm{This manuscript contributes to this area
  of research by examining a new concept of guaranteed privacy and its
  application in dynamic system estimation. }  


{\emph{{Literature Review}.}}  Numerous information-theoretic notions
have been proposed to measure the concept of privacy, and these
definitions can be put into practice when dealing with the analysis of
real-time data streams \cite{sankar2013role}. A main approach to this
is \emph{differential privacy}~\cite{dwork2006calibrating}, originally
proposed for the protection of databases of individual records subject
to public queries. A system handling sensitive inputs achieves
differential privacy through the randomization of its responses. This
randomization is carefully designed to ensure that the distribution of
publicly disclosed outputs remains relatively insensitive to the data
contributed by any individual participant. This concept has been
broadened and applied across various domains, including machine
learning and
regression~\cite{chaudhuri2011differentially,zhang2012functional,hall2013differential},
control, estimation, and
verification~\cite{wang2017differential,han2021numerical}, multi-agent
systems (consensus, message
passing)~\cite{huang2015differentially,han2016differentially,hale2015differentially},
as well as optimization and
games~\cite{ding2021differentially,li2020privacy,ye2021differentially}.
  
Considering dynamic settings, differential privacy has been applied to
filtering, assuming either that the statistical characterizations of
uncertainties are known \cite{le2013differentially} or that there is
no disturbance \cite{le2020differentially}. However, these approaches
are not applicable to bounded-error settings where uncertainties are
only assumed to be bounded (set-valued) with unknown distributions. In
such settings, interval observers
\cite{zheng2016design,khajenejad2022h} are capable of providing
guaranteed and uniformly bounded state estimates. The work in
\cite{degue2020differentially} proposed a differentially private
mechanism to augment an existing interval observer \mkj{for LTI systems}. 
  This was done via an input perturbation
mechanism, by which stochastic bounded-support noise was added to each
individual's data prior to sending it to the observer. 
The existence of such initial stable observer (i.e, a stabilizing
gain) was assumed to be granted. \mkj{Moreover,} after the injection of the additional
stochastic perturbation, neither the \mkj{correctnes, i.e., the framer property of the} observer, 
nor its stability \mkj{were} re-evaluated. 
While~\cite{degue2020differentially} provide\mkj{d} a
first design method that is  inclusive of differential privacy, the question of guaranteed-\mkj{private}-stable and optimal design was left unaddressed, which this paper contributes toward.

{\emph{{Contributions}.}} 
We start by refining a new notion of guaranteed privacy, which characterizes privacy in terms of how close the ranges of two set-valued estimates of the published data are, \mmoh{i.e, how small the distance of the set-estimates and how big their intersection are}. \mmoh{As opposed to stochastic differential privacy, guaranteed privacy is deterministic, i.e., provides hard and quantifiable upper bounds for the distance between \emph{any} two possible values belonging to the guaranteed set of estimates of the published data.} Then, we synthesize an interval observer \mmoh{for a perturbed LTI bounded-error system}, though designing \mkj{a} bounded noise perturbation factor, as well as observer gain\mmoh{s}. \mmok{Our observer design, which is \emph{correct-by-construction}, returns guaranteed set-valued estimates of a desired variable, which is a function/transformation of the system state. This allows us to reserve the choice of observer gains and perturbation factor to satisfy both stability and privacy.} \mkj{Hence, the synthesized observer} simultaneously returns guaranteed private and stable interval-valued estimates of \mmoh{the} desired variable. Further, we show that our design is optimal, in the sense that it minimizes the $\mathcal{H}_{\infty}$ norm of the observer error system. Finally, we study the accuracy of our proposed mechanism by quantifying the loss due to considering guaranteed privacy specifications. \mmoh{This results in designing guaranteed privacy-preserving observers that are optimal
in the sense that they minimize the corresponding inaccuracy due to the perturbations. Furthermore, the designed mechanisms are robust against uncertainties in environment, in the sense of considering the worst-case scenario. \mmok{Finally, our design approach is based on solving a mixed-integer semi-definite programs (MISDP), with the main advantage that it does not require us to impose additional constraints, and hence conservatism, as is the case for the semi-definite programming (SDP)-based solutions. In other word, our proposed solution is \emph{tight}.}} 

\section{Preliminaries}
In this section, we introduce basic notation, as well as preliminary
concepts and results used in the sequel.

{\emph{{Notation}.}}  $\mathbb{R}^{\moo{m}},\mathbb{R}^{\moo{m} \times p}$, and
$\mathbb{R}^\moo{m}_{\geq 0}$ denote the $\moo{m}$-dimensional Euclidean space,
the sets of $\moo{m}$ by $p$ matrices, and nonnegative vectors in
$\mathbb{R}^\moo{m}$, respectively.  Also $\mathbf{1}_{\moo{m}},\mathbf{0}_{\moo{m}}$, and
$\mathbf{0}_{\moo{m}\times p}$ denote the the vectors of ones and zeros in
$\mathbb{R}^{\moo{m}}$, and the matrix of ones in $\mathbb{R}^{\moo{m} \times p}$,
respectively. Further, \mmoh{$\mathbb{R}^{n \times n}_{\succ 0}$ and} $\mathbb{D}^{{\moo{m}} \times {\moo{m}}}_{>0}$ denote the
spaces of $n$ by $n$ \mmoh{positive definite and positive} diagonal matrices, \mmoh{respectively}. Given
$M \in \mathbb{R}^{\moo{m} \times p}$, $M^\top$ represents its transpose,
$M_{ij}$ denotes $M$'s entry in the $\supscr{i}{th}$ row and the
$j^\textup{th}$ column,
$M^{\oplus}\triangleq \max(M,\mathbf{0}_{\moo{m} \times p})$ \mmoh{(component-wise)},
$M^{\ominus}=M^{\oplus}-M$, $|M|\triangleq
M^{\oplus}+M^{\ominus}$, \moo{and $\sigma_{\max}(M) \triangleq \max_{x} \|Mx\|_2$ s.t. $\|x\|_2 \triangleq \sum_{i=1}^p x_i^2=1$ denotes the maximum singular value of $M$}. Furthermore, for
$a,b \in \mathbb{R}^n, a\leq b$ means
$a_i \leq b_i, \forall i \in \until{\moo{m}}$, while
$\tdiag (D^1,\dots,D^N)$ denotes a block diagonal matrix with diagonal
blocks $D^1,\dots,D^N$. \mkj{Further}, a function
$\theta:\mathbb{R}_{\geq 0} \to \mathbb{R}_{\geq 0}$ is of class
$\mathcal{K}$ (resp.~$\mathcal{K}_{\infty}$) if it is
continuous, 
and strictly increasing (resp.~if is of class $\mathcal{K}$ and also
unbounded). Moreover,
$\kappa:\mathbb{R}_{\geq 0} \to \mathbb{R}_{\geq 0}$ if of class
$\mathcal{KL}$ if for each fixed $k \geq 0$, $\kappa(.,t)$ is of class
$\mathcal{K}$ and for each $s\geq 0$, $\kappa(s,t)$ decreases to zero
as $t \to \infty$. \moo{Finally, given any arbitrary sequence $\{s_k\}_{k=0}^{\infty}$, $\|s\|_{\ell_2} \triangleq \sqrt{\Sigma_{k=0}^{\infty}\|s_k\|^2_2}$ and $\|s\|_{\ell_{\moo{\infty}}} \triangleq \sup_{k \in
  \mathbb{K}}\|s_k\|_2$ denote its $\ell_2$ and $\ell_{\infty}$ signal norms.}
\begin{defn}[Intervals]\label{def:interval}
  \mh{An $\moo{m}$-dimensional interval {$\mathcal{I} \triangleq
      [\underline{\md{z}},\overline{\md{z}}] \subset 
      \mathbb{R}^{\moo{m}}$}, is the set of
    all real vectors $\md{z \in \mathbb{R}^{\moo{m}}}$ that satisfy
    $\underline{\md{z}} \le \md{z} \le \overline{\md{z}}$. {Moreover,} 
    $\textstyle{\mathrm{diam}}(\mathcal{I}) \triangleq
    \|\overline{\md{z}}-\underline{\md{z}}\|\mk{_{\infty}\triangleq
      \max_{i \in
        \{1,\cdots,\md{\moo{m}}\}}\md{|\overline{z}_i-\underline{z}_i|}}$
    denotes the {diameter} or {interval width} of
    $\mathcal{I}$, \mkj{while $c \triangleq \frac{\underline{z}+\overline{z}}{2}$ is the center of $\mathcal{I}$}. Finally, $\mathbb{IR}^n$ denotes the space of all
    $n$-dimensional intervals, also referred to as \emph{interval
      vectors}. } 
    \end{defn}
\mmoh{Note that for a general convex set $\mathcal{S}$, $\tdiam (\mathcal{S})$ denotes the diameter of $\mathcal{S}$, i.e., $\max_{q,q'\in \mathcal{S}} \|q-q'\|_{\infty}$.}
\begin{prop}\cite[Lemma 1]{efimov2013interval}\label{prop:bounding}
  Let $A \in \real^{p \times n}$ and $\underline{x} \leq x \leq
  \overline{x} \in \real^n$. Then, $A^+\underline{x}-A^{-}\overline{x}
  \leq Ax \leq A^+\overline{x}-A^{-}\underline{x}$. As a corollary, if
  $A$ is non-negative, $A\underline{x} \leq Ax \leq A\overline{x}$.
\end{prop}
\section{Problem Formulation} \label{sec:Problem} 
\subsection{System Assumptions}
Consider a set of $N$ linear time invariant discrete-time
bounded-error systems (agents) with the following dynamics:
\begin{align}\label{eq:system_ind}
\begin{array}{rl}
x^i_{k+1}&=A^ix^i_k+\sum_{j\neq i}A^{ij}x^j_k+W^iw^i_{k},\\
y^i_{k}&=C^ix^i_k+V^iv^i_{k},
\end{array}
\end{align}
where 
$k \in \K \triangleq \mathbb{N} \cup \{0\}, i \in \{1,\dots,
N\}, x^i_0 \in [\underline{x}^i_0,\overline{x}^i_0]$,
$x^i_k \in \real^{n^i}$ is the state vector of the agent $i$ and
$w^i_k \in \mathcal{I}^i_{w} \triangleq
[\underline{w}^i,\overline{w}^i] \subset \real^{n^i_{w}}$ is a bounded
process disturbance. Furthermore, at time step
$k$, every system (agent) $i$ takes (originates) a distinct
privacy-sensitive vector-valued measurement signal
$y^i_k \in \real^{m^i}$, which is affected by
$v^i_k \in \mathcal{I}_{v}^i \triangleq
[\underline{v}^i,\overline{v}^i] \subset \real^{n_v^i}$, a bounded
sensor (measurement) noise signal.
Finally,
$A^i \in \real^{n^i \times n^i},W^i \in \real^{n^i \times n^i_w},C^i
\in \real^{m^i \times n^i}$ and $V^i \in \real^{m^i \times n^i_v}$ are
known constant matrices, while $A^{ij}, j \neq i,$ represent coupling
matrices that capture the influence of the other agents on the agent
$i$. Unlike to the work in \cite{degue2020differentially}, we do not
impose any restrictions on $W^i$.  The global system dynamics can be
constructed by the agents and with
$\moo{n}\triangleq \sum_{i=1}^N n^i$ 
  states and $\moo{m}\triangleq \sum_{i=1}^N m^i$ outputs, as the following plant:
\begin{align}\label{eq:system}
\mathcal{G}: \begin{cases}x_{k+1}&=Ax_k+Ww_{k},\\
y_{k}&=Cx_k+Vv_{k} \end{cases},
\end{align}
where $x_0$ is unknown, but satisfies
$\underline{x}_0 \leq x_0 \leq \overline{x}_0$. Moreover,
\begin{align*}
  \xi_k &\triangleq [(\xi^1_k)^\top,\dots (\xi^N_k)^\top]^\top, \ \forall \xi \in \{x,y,w,v\},\\
  J &\triangleq \tdiag (J^1,\dots,J^N), \quad \forall J \in \{W,C,V\}, \\
A_{ii}&\triangleq A^i, A_{ij}=A^{ji}, \forall j \neq i, \forall i \in \{1,\dots,N\},
\end{align*}
and the data $\underline{x}_0,\overline{x}_0, A,C,W,V$ are assumed to
be public information. Furthermore, the bounded general state and
measurement noise signals satisfy
$\underline{\xi} \leq \xi_k \leq \overline{\xi},\forall \xi \in
\{w,v\}$, where
$\underline{\xi}\triangleq [(\underline{\xi}^1_k)^\top,\dots
(\underline{\xi}^N_k)^\top]^\top,\overline{\xi} \triangleq
[(\overline{\xi}^1_k)^\top,\dots (\overline{\xi}^N_k)^\top]^\top$.

{\color{black} There is an operator, whose objective is to obtain
  interval-valued estimates of $z_k$ given by:
  $z_k \triangleq \Gamma x_k = \textstyle{\sum}_{i=1}^N \Gamma^i
  x^i_k$, where $\Gamma^i$ can be any arbitrary matrices.  \mmoh{The
    variable $z_k$ represents a special output of the system, e.g., a
    selection (subset) of the individual states or outputs or a linear
    combination of these, which can be used for subsequent decisions
    including future agents' actions (cf. Section~\ref{sec:example}
    for a specific example concerning a dynamic market).}  The
  estimates of $z_k$ should satisfy
  $\underline{z}_k \leq z_k \leq \overline{z}_k$ and be as tight as
  possible. 
 
  To do this, the operator employs interval estimates of $x_k$, which
  in turn depend on the system data outputs $y_k$. The goal of the
  operator is to ensure that the publicly released interval estimates
  of $z_k$, $\mathcal{Z}_k \triangleq [\underline{z}_k , \underline{z}_k ] \equiv \mathcal{M}(y_k)$,
  which are functions of the data $y_k$,
  \sm{ensure} the privacy of the
  data in a guaranteed manner.} The rationale behind this concern
lies in the potential for extracting fresh insights about the
multi-agent system 
through interval estimates of $z_k$. 
 This can be accomplished, for instance, by exploiting linkage
 attacks. In such scenarios, an adversary \sm{can deduce novel
   information about}
 specific individuals \cite{sweeney2002achieving} by combining the
 newly published information with \sm{additional side knowledge}.
 
 \mmoh{Hence, to address this concern,} the operator aims to satisfy a deterministic notion
 of privacy, 
  \sm{which ensures} that
  the publicly released $\underline{z}_k,\overline{z}_k$ guarantee
  \emph{hard (deterministic) privacy bounds} for each agent's data.
  \mmoh{To achieve this, the operator/system designer perturbs the state equation as well as the
    output signal $y_k$, \mmok{and consequently $z_k$,} by some
    intentionally added bounded noise/disturbance, and a 
    perturbation factor is being designed such that the set-valued estimates of $z_k$ corresponding to two close enough, i.e., adjacent outputs $y_k$ and $y'_k$, are hard to be distinguished by the adversary.}     


The synthesis of the interval-valued estimates is being done through
an $\mathcal{H}_{\infty}$-optimal interval observer, which is formally
introduced via the following sequence of definitions.
\begin{defn}[Interval Framer]\label{defn:framer}
  The sequences $\{\underline{x}_k,\overline{x}_k\}_{k=0}^{\infty}$
  are called lower and upper framers for the states of system
  $\mathcal{G}$ if
  $\forall w_k \in [\underline{w},\overline{w}], \forall v_k \in
  [\underline{v},\overline{v}], \underline{x}_k \leq x_k \leq
  \underline{x}_k$. Further, any dynamical system $\widehat{\mathcal{G}}$
  whose states are framers for the states of $\mathcal{G}$, i.e., any
  (tractable) algorithm that returns framers for the
  states of \eqref{eq:system}, is called an interval framer for $\mathcal{G}$.
\end{defn}
\begin{defn}[Input-to-State Stability \& Interval Observer]\label{defn:stability}
  An interval framer $\widehat{\mathcal{G}}$ is input-to-state stable
  (ISS), if the framer error
  $e^x_k \triangleq \overline{x}_k-\underline{x}_k$ is bounded as
  follows:
\begin{align}\label{eq:ISS}
  \|e^x_k\|_2
  \leq \kappa (\|e^x_0\|_2,k)+\theta(\|\delta_{\tilde{w}}\|_{\ell_{\moo{\infty}}}) \quad
  \forall k \in \mathbb{K},
\end{align}
where
\mmoh{$\delta_{\tilde{w}} \triangleq [\delta_{w}^\top \
\delta_{v}^\top]^\top \triangleq [(\overline{w}-\underline{w})^\top (\overline{v}-\underline{v})^\top]^\top$}
   \moo{is} the augmented vector of
 noise widths, \moo{while} $\kappa$ and $\theta$ are functions of classes
 $\mathcal{KL}$ and $\mathcal{K}_{\infty}$, respectively.
An ISS interval framer is called an
interval observer.
\end{defn}
 \begin{defn}[$\mathcal{H}_{\infty}$-Optimal Interval Observer]\label{defn:opt_obs}
   An interval observer $\widehat{\mathcal{G}}$ is
   $\mathcal{H}_{\infty}$-optimal if the $\mathcal{H}_{\infty}$ gain
   of the \sm{dynamic system in the state estimation error}
   $\moo{\widetilde{\mathcal{G}}}$, 
     i.e.,
   $\|\moo{\widetilde{\mathcal{G}}}\|_{\mathcal{H}_{\infty}} \triangleq
   \sup_{\delta_{\tilde{w}\ne 0}}
   \frac{\|e^x\|_{\ell_2}}{\|\delta_{\tilde{w}}\|_{\ell_2}}$, is
   minimized, with its minimum value called the observer's noise
   attenuation level. 
 \end{defn}
 Note that we do not assume that such an observer exists, 
 rather we want to synthesize it while satisfying privacy
 specifications with deterministic bounds. This is formalized via the
 notion of \emph{guaranteed privacy} as described next.
 
\subsection{Guaranteed Privacy} 
To formally define guaranteed privacy, we use a 
\moo{version} of \emph{adjacency relation}. 
  Let $\mathcal{Y}$ denote the space of
measured signal sequences $\{y_k\}_{k \geq 0}$, and $\rho >0$
be given. A symmetric binary relation on
$\mathcal{Y}$, denoted \moo{$\adj_{\rho}$}, identifies the types of
variations in $y$ that we aim to make hard to detect. 
\begin{defn}[$\rho$ Adjacency Relation]\label{def:adjacency}
  \mmok{Given $\rho >0$, and} any $y,y' $ \mmok{in a vector space
    $\mathcal{Y}$,}
  \begin{align}\label{eq:adj}
    \moo{\adj_{\rho}} (y,y') \ \text{if and only if} \ \|y-y'\|_{\moo{2}} \leq \rho.
  \end{align}
\end{defn}
Such interpretation of adjacent datasets implies that a single
participant \sm{possibly} contributes additively to each $y^i$ in a
way that its overall impact on the dataset $y$ is bounded in
\moo{$2$}-norm by
$\rho$. 
\mmok{Before formally introducing the notion of guaranteed privacy, we provide some insights about what we aim to achieve by satisfying this notion of privacy in bounded-error settings. Our goal is to ensure that the set-valued estimator mapping (or mechanism), i.e., 
\begin{align*}
\mathcal{M}(y_k)=\mathcal{Z}_k \triangleq [\underline{z}_k,\overline{z}_k]
\end{align*}
is private in the sense that an adversary cannot distinguish between $\rho$-adjacent data as they produce very close outputs (i.e.~the
  inversion of $\mathcal{M}$ leads to a large dataset). To ensure this, we design the mechanism $\mathcal{M}$ such that after some large enough time step and for any two arbitrary $\rho$-adjacent outputs $y,y'$ and their corresponding set-valued estimates $\mathcal{Z}_k=\mathcal{M}(y_k),\mathcal{Z}'_k=\mathcal{M}(y'_k)$:
 \begin{enumerate}[(i)]
 \item the sets $\mathcal{Z}_k$ and $\mathcal{Z}'_k$ remain close to each other, i.e., the distance between the two sets is small, and \label{privacy:distance}
 \item there is a large enough overlap (intersection) between $\mathcal{Z}_k$ and $\mathcal{Z}'_k$. \label{privacy:intersection} 
 \end{enumerate}
 These two properties, which are formalized through the following definition, together imply that the sets $\mathcal{Z}_k$ and $\mathcal{Z}'_k$ are hard to be distinguished by the adversary, hence ensuring that the estimator mapping
  $\mathcal{M}(y_k)$ is private.} 

\begin{defn}[Guaranteed Privacy]\label{defn:guar_privacy}
  Let $\epsilon, \delta \geq 0$, and $D$ be \moo{vector space} equipped
  with the symmetric binary relation $\moo{\adj_{\rho}}$ given in
  Definition \ref{def:adjacency}. A deterministic set-valued mechanism
  $\mathcal{M}: D \to \mathbb{IR}^n$ is \mmoh{weakly $\epsilon$-guaranteed
  private} w.r.t.~$\moo{\adj_{\rho}}$, 
  if for all $d,d' \in D$ such
  that $\moo{\adj_{\rho}}(d,d')$, all $q \in \mathcal{M}(d)$, and all
  $q' \in \mathcal{M}(d')$, we
  have 
  \mmoh{
\begin{align}
\label{eq:GP_distance}&\|q-q'\|_2 \leq e^{\epsilon}.
\end{align}
Moreover, $\mathcal{M}$ is strongly $(\epsilon,\delta)$-guaranteed private w.r.t.~$\moo{\adj_{\rho}}$ if it is weakly $\epsilon$-guaranteed
 private and the following holds:
\begin{align}
\label{eq:GP_overlap}&\ttdiam (\mathcal{M}(d) \cap \mathcal{M}(d')) \geq \delta,
\end{align}
}
where $\mathcal{M}(d)$ and $\mathcal{M}(d')$ are the entire interval
range of $\mathcal{M}$ applying \mkj{to} $d$ and $d'$, respectively.
\end{defn}
\mmoh{The inequality in \eqref{eq:GP_distance} ensures condition \eqref{privacy:distance} above, while the one in \eqref{eq:GP_overlap} implies \eqref{privacy:intersection}. So, if these two together hold for sufficiently small $\epsilon$ and large enough $\delta$, then an adversary can hardly distinguish between $\rho$-adjacent data as they produce close outputs, i.e.,~the
  inversion of $\mathcal{M}$ leads to a large dataset.} 
  
 It is worth mentioning that this notion
of guaranteed privacy is stronger than the one introduced
in~\cite[Definition 3]{khajenejad2022guaranteed} 
for distributed nonconvex optimization. 
In other words, Definition~\ref{defn:guar_privacy} implies the one in
\cite{khajenejad2022guaranteed}, but not conversely.  \mkj{Further},
we re-emphasize the difference between this notion with that of
differential privacy
in~\cite{nozari2016differentially,ding2021differentially,li2020privacy}.
Under differential privacy, the statistics of the output of
$\mathcal{M}$, i.e., the probability of the values of $\mathcal{M}$,
  is allowed to change only slightly if there is a slight
perturbation of the data $y$.  Instead, when guaranteed privacy is
considered, the entire range of the set-valued mechanism $\mathcal{M}$
is allowed to change only slightly with respect to the perturbed data.


With this being said, our problem can be cast as follows:
\begin{problem}[Guaranteed Privacy-Preserving Interval Observer
  Design]\label{pro:GPPIOD}
  Given system \moo{$\mathcal{G}$}, 
    design a mechanism (or mapping)
  $\mathcal{M}$ that simultaneously
\begin{itemize}
\item  outputs framers for $z_k$ through a to-be-designed framer system (cf. Definition \ref{defn:framer}), 
\item   ensures that the framer system is ISS, i.e., the framer system is an interval observer (cf. Definition \ref{defn:stability}), 
\item satisfies the guaranteed privacy of data $\{y_k\}_{k=0}^{\infty}$ (cf. Definition \ref{defn:guar_privacy}), and 
\item  guarantees that the observer design is optimal in the sense of $\mathcal{H}_{\infty}$ (cf. Definition \ref{defn:opt_obs}). 
\end{itemize} 
\end{problem}

\section{Guaranteed Privacy-Preserving Interval Observer
  Design} \label{sec:observer} In this section, we introduce our
proposed strategy to design a guaranteed privacy-preserving mechanism
(or mapping) for interval observer design, addressing
Problem~\ref{pro:GPPIOD}. {Our proposed approach is a \mmoh{threefold}
  procedure. Before describing the technical details of such procedure in the following subsections, below, we provide a brief intuitive overview of all the steps.  
  
   \textbf{(i)} First, we perturb system $\mathcal{G}$ by
  injecting additional set-valued output noise \mmoh{and state disturbance}, represented by the
  to-be-designed perturbation factor $\alpha >0$, to ``hide" the true value of the data and
  prevent it to be revealed by the adversary. 
  We assume the additional bounded perturbation noise and disturbance, $v^a_k$ and $w^a_k$, satisfy:
  \mmoh{\begin{align*}
  &\underline{\hat{v}} \triangleq \alpha \underline{v} \leq v^a_k+v_k \leq \alpha \overline{v} \triangleq \overline{\hat{v}} ,\\
   & \underline{\hat{w}} \triangleq \alpha \underline{w} \leq w^a_k+w_k \leq \alpha \overline{w} \triangleq \overline{\hat{w}},
   \end{align*}}
  where $\alpha$ 
  will be designed later (along with other design factors) to satisfy
  the desired properties of the mechanism. After injecting the
  controlled \mmoh{perturbations $v^a_k$ and $w^a_k$, we obtain the following perturbed system:
  \begin{align}\label{eq:system_perturbed}
\mathcal{G}_p: \begin{cases}x_{k+1}&=Ax_k+W\hat{w}_{k}, \quad \hat w_k \in [\underline{\hat{w}},\overline{\hat{w}}],\\
y_{k}&=Cx_k+V\hat{v}_{k},\quad \ \ \hat v_k \in [\underline{\hat{v}},\overline{\hat{v}}], \\
z_k&=\Gamma x_k.  \end{cases}
\end{align}} 
\mmoh{We then propose an interval observer for $\mathcal{G}_p$, i.e., a set-valued mapping/mechanism $\mathcal{M}$, that outputs stable intervals $\mathcal{Z}_k=\mathcal{M}(y_k)$, that are guaranteed to contain true values of $z_k$, given output $y_k$}.
  
  \textbf{(ii)} \mmoh{Next}, we provide an upper bound for the distance between $z_k \in\mathcal{Z}_k=\mathcal{M}(y_k)$ and
  $ z'_k \in \mathcal{Z}'_k=\mathcal{M}(y'_k)$, \mmoh{corresponding to sufficiently close $y_k$ and $y'k$}, by upper bounding the
  distance between the centers of the intervals $\mathcal{Z}_k$ and $\mathcal{Z}'_k$. \mmoh{This results in sufficient conditions to guarantee that $\mathcal{Z}_k$ and $\mathcal{Z}'_k$ remains close to each other.}  
  
  \mmoh{\textbf{(iii)} Finally, we ensure that the summation of the radii of $\mathcal{Z}_k$ and $\mathcal{Z}'_k$ remain bigger enough than the distance between their centers. This implies that the two intervals have a large enough intersection.} 
  
 Note that \textbf{(i)}--\textbf{(iii)} together result in a guaranteed privacy-preserving stable and optimal interval observer design, where all these specifications are satisfied simultaneously using \mmoh{six} degrees of
  freedom: \mmoh{$\alpha, Q, P, \tilde{L}, \tilde{N}, \tilde{T}$}, designed through Theorems
  \ref{thm:privacy} and \ref{thm:synthesis}.}
\subsection{Framer Structure}
\mmoh{We start by proposing an interval observer for the perturbed system $\mathcal{G}_p$, i.e., a dynamical system that outputs correct and stable framers for $z_k$. To do so, we first define an auxiliary state $\xi_k$: 
\begin{align}\label{eq:aux_state}
\xi_k \triangleq x_k-N(y_k-V\hat v_k)=(I-NC)x_k=Tx_k,
\end{align}
where $T$ and $N$ are to-be-designed observer gains satisfying
\begin{align}\label{eq:add_constraint}
T+NC=I.
\end{align}
Then, from \eqref{eq:system_perturbed}, \eqref{eq:aux_state} and the fact that $L(y_k-Cx_k-V\hat v_k=0)$ for any observer gain $L \in \mathbb{R}^{n \times m}$, we obtain the following equivalent representation for the dynamics of $\mathcal{G}_p$:  
\begin{align}\label{eq:equivalent_system}
\hspace{-.3cm}\begin{cases}
\xi_{k+1}\hspace{-.1cm}=\hspace{-.1cm}Tx_{k+1}\hspace{-.1cm}=\hspace{-.1cm}T(Ax_k\hspace{-.1cm}+\hspace{-.1cm}W\hat w_k)\hspace{-.1cm}+\hspace{-.1cm}L(y_k\hspace{-.1cm}-\hspace{-.1cm}Cx_k\hspace{-.1cm}-\hspace{-.1cm}V\hat v_k)\\
=(TA-LC)x_k+TW\hat w_k-TV\hat v_k+Ly_k\\
=\hspace{-.1cm}(TA\hspace{-.1cm}-\hspace{-.1cm}LC)\xi_k\hspace{-.1cm}+\hspace{-.1cm}TW\hat w_k\hspace{-.1cm}-\hspace{-.1cm}(T\hspace{-.1cm}+\hspace{-.1cm}(TA-LC)N)V\hat v_k\\
+(L+(TA-LC)N)y_k,\\
x_k=\xi_k+N(y_k-V\hat v_k).
\end{cases}
\end{align}
} 
 \mmoh{
Next, we propose the following framer system for \eqref{eq:equivalent_system}:
\begin{align}\label{eq:observer}
{\widehat{\mathcal{G}}}:\begin{cases} 
\underline{\xi}_{k+1}&\hspace{-.2cm}=(TA\hspace{-.1cm}-\hspace{-.1cm}LC)^\oplus \underline{\xi}_k\hspace{-.1cm}-\hspace{-.1cm}(TA-LC)^\ominus \overline{\xi}_k+\tilde y_k\\
&\hspace{-.2cm}+(TW)^\oplus\underline{\hat w}-(TW)^\ominus\overline{\hat w}+Z^\ominus\underline{\hat v}-Z^\oplus\overline{\hat v} ,\\
\overline{\xi}_{k+1}&\hspace{-.2cm}=(TA\hspace{-.1cm}-\hspace{-.1cm}LC)^\oplus \overline{\xi}_k\hspace{-.1cm}-\hspace{-.1cm}(TA-LC)^\ominus \underline{\xi}_k+\tilde y_k\\
&\hspace{-.2cm}+(TW)^\oplus\overline{\hat w}-(TW)^\ominus\underline{\hat w}+Z^\ominus\overline{\hat v}-Z^\oplus\underline{\hat v} ,\\
\underline{x}_{k}&\hspace{-.2cm}=\underline{\xi}_k+(NV)^{\ominus}\underline{\hat v}-(NV)^{\oplus}\overline{\hat v}+Ny_k ,\\
\overline{x}_{k}&\hspace{-.2cm}=\overline{\xi}_k+(NV)^{\ominus}\overline{\hat v}-(NV)^{\oplus}\underline{\hat v}+Ny_k ,\\
\underline{z}_{k}&\hspace{-.2cm}=\Gamma^{\oplus}\underline{x}_k-\Gamma^{\ominus}\overline{x}_k,\ \overline{z}_{k}=\Gamma^{\oplus}\overline{x}_k-\Gamma^{\ominus}\underline{x}_k,
\end{cases}
\end{align}
initialized at $[\overline{x}^\top_0,\underline{x}^\top_0]^\top$, where $Z \triangleq (T+(TA-LC)N)V$ and $\tilde{y}_k \triangleq (L+(TA-LC)N)y_k$. Here, $\underline{x}_k,\overline{x}_k,\underline{z}_k,\overline{z}_k$ can be
  interpreted as the ``outputs" of $\widehat{\mathcal{G}}$.}
The following proposition shows that \mkj{$\widehat{\mathcal{G}}$} \mmoh{indeed} constructs a framer for
\mmoh{$\mathcal{G}_p$} for all values of \mmoh{the observer gains}. \mmok{This allows us to reserve the choice
  of \mmoh{$L,N,T$} and $\alpha$ to satisfy both stability and privacy later (cf. Theorems \ref{lem:stability}--\ref{thm:synthesis}).}
\begin{prop}[Framer Property]\label{prop:framer}
  The state trajectory of system \eqref{eq:observer}, initialized at
  $[\overline{x}^\top_0,\underline{x}^\top_0]^\top$, frames the true
  state of \eqref{eq:system} at each time step $k$, i.e.,
  $\underline{x}_k \leq x_k \leq \overline{x}_k, \ \forall k \geq 0$. Moreover, $\forall k \geq 0, \underline{z}_k \leq z_k =\Gamma_k 
  \leq \overline{z}_k$, i.e., $z_k \in [\underline{z}_k,\overline{z}_k]$.
\end{prop}
\begin{proof}
  Starting from \mmoh{\eqref{eq:equivalent_system} and applying Proposition \ref{prop:bounding} to all multiplications of matrices with
  uncertain vectors yields the dynamics of the auxiliary variable, i.e., the first two equations in \eqref{eq:observer}. Then, from \eqref{eq:aux_state} we have $ x_k=\xi_k+Ny_k-NV_k$, where applying Proposition \ref{prop:bounding} again results in the third and fourth equations in \eqref{eq:observer}. 
  Finally, the last two equations follow follow from applying Proposition \ref{prop:bounding} to $z_k=\Gamma x_k$}. 
  \end{proof}

\subsection{Observer Input-to-State Stability}
In this subsection, we formalize sufficient conditions to satisfy the
stability of the proposed observer~\eqref{eq:observer}. These conditions
simultaneously satisfy some bounds on the $\mathcal{H}_{\infty}$-norm
of the observer error dynamics, which is also required in the next
subsection where we provide guaranteed privacy-preserving conditions. \mmoh{We begin by computing a linear comparison system for the error dynamics.
\begin{lem}\label{lem:error_dynamics}
The observer error dynamics admits the folllwoing linear positive comparison system:
\begin{align}\label{eq:error}
\moo{\widetilde{\mathcal{G}}}: e^x_{k+1}=|{TA-LC}|e^x_k+\Lambda \delta_{\hat \lambda}, \ \text{where}
\end{align} 
\begin{align*}
  \Lambda \triangleq [|TW| \quad |LV|+|NV|], \ \delta_{\hat \lambda}=\alpha \delta_{\lambda}
  \triangleq \alpha[\delta^\top_w \ \delta^\top_{v}]^\top.
\end{align*}
\end{lem}
\begin{proof}
Starting
from 
%
\eqref{eq:observer}, and defining
\mmoh{$ \delta_{\hat s}=\alpha \delta_{s} \triangleq \alpha (\overline{s}-\underline{s}),\forall s
\in \{w,v\}$, as well as the observer errors
$e^r_k \triangleq \overline{r}_k-\underline{r}_k, r\in \{x,\xi, z\}$}, it is straightforward to obtain the following: 
\begin{align}
\label{eq:error_1} e^\xi_{k+1}&=|TA-LC|e^{\xi}_k+|TW|\delta_{\hat w}+|Z|\delta_{\hat v},\\
 \label{eq:error_2}e^x_k&=e^\xi_k+|NV|\delta_{\hat v},\\
 \label{eq:error_3}e^z_k&=|\Gamma|e^x_k.
\end{align}
Next, combining \eqref{eq:error_1} and \eqref{eq:error_2} yields the following dynamical system and its corresponding comparison system $\widetilde{\mathcal{G}}$ for the evolution of $e^x_k$: 
\begin{align}\label{eq:error_comparison}
\begin{array}{rl}
&e^x_{k+1}= |TA-LC|e^x_k+|TW|\delta_{\hat w}\\
&+(|(T\hspace{-.1cm}+\hspace{-.1cm}(TA\hspace{-.1cm}-\hspace{-.1cm}LC))NV|\hspace{-.1cm}+\hspace{-.1cm}|NV|\hspace{-.1cm}-\hspace{-.1cm}|TA\hspace{-.1cm}-\hspace{-.1cm}LC||NV|)\delta_{\hat v}\\
&\leq |TA-LC|e^x_k+|TW|\delta_{\hat w}+(|TV|+|NV|)\delta_{\hat v},
\end{array}
\end{align}
where the inequality follows from the sub-multiplicative property of the absolute value operator.
\end{proof}}
\mmoh{Note that the observer error dynamics and its comparison system $\widetilde{\mathcal{G}}$ are discrete-time positive systems by construction. Using this fact,} the following \mmoh{theorem} provides \mmoh{necessary and} sufficient conditions \mkj{in the form a mixed-integer semi-definite program (MISDP)} for the stability of the comparison system. \mmoh{Moreover, Theorem \ref{lem:stability} provides} an upper bound for the
$\mathcal{H}_{\infty}$-norm \mmoh{of $\widetilde{\mathcal{G}}$, which will be used later in Theorem \ref{thm:privacy} to derive sufficient conditions for privacy.} 
{
\begin{thm}[Observer Stability and Error Dynamics Upper Bound]\label{lem:stability}
  \mmoh{The error comparison system $\widetilde{\mathcal{G}}$ is ISS and satisfies 
   \begin{align}\label{eq:error_norm}
    \|\moo{\widetilde{\mathcal{G}}}\|_{\mathcal{H}_{\infty}} \triangleq \sup\limits_{k \geq 0} \frac{\|e^{x}_k\|_2}{\|\delta_{\hat \lambda}\|_2} \leq \mmoh{\gamma},
\end{align}
 if and only if there exists $(\gamma_*,P_*,\widetilde{L}_*,\widetilde{N}_*,\widetilde{T}_*)$ that solves the following MISDP: 
  \begin{align}\label{eq:stability}
\begin{bmatrix} P & |\widetilde{T}A\hspace{-.1cm}-\hspace{-.1cm}\widetilde{L}C| &\begin{bmatrix}|\widetilde{T}W| &  |\widetilde{L}V|\hspace{-.1cm}+\hspace{-.1cm}|\widetilde{N}V|\end{bmatrix} & 0 \\
                           * & P                                                   &  0                                                                                                     &  I \\
                           * & *                                                     & \gamma I                                                                                         & 0 \\
                           * & *                                                      & *                                                                                                      & P 
\end{bmatrix}\hspace{-.1cm} \succ \hspace{-.1cm}0, \tilde{T}\hspace{-.1cm}+\hspace{-.1cm}\tilde{N}C\hspace{-.1cm}=\hspace{-.1cm}P,
\end{align}
where 
  $P \in \mathbb{D}^{{n} \times {n}}_{>0}$, $\widetilde{T} \in \mathbb{R}^{n \times n}$, $\gamma > 0$, and
  ${\widetilde{L}},{\widetilde{N}} \in \mathbb{R}^{n \times m}$.} 
\end{thm}
\begin{proof}
First, note that by construction, $\widetilde{\mathcal{G}}$ is a positive discrete-time system. Hence, \mmoh{it follows from \cite[Corollary 4]{ebihara2014lmi} as well as applying Schur complement} that $\widetilde{\mathcal{G}}$ is stable and satisfies \eqref{eq:error_norm} \mmoh{if and only if} there exists a positive diagonal matrix $P$ such that  
\begin{align}\label{eq:stability_1}
\begin{bmatrix} P & P|{T}A\hspace{-.1cm}-\hspace{-.1cm}{L}C| &P\begin{bmatrix}|{T}W| &  |{L}V|\hspace{-.1cm}+\hspace{-.1cm}|{N}V|\end{bmatrix} & 0 \\
                           * & P                                                   &  0                                                                                                     &  I \\
                           * & *                                                     & \gamma I                                                                                         & 0 \\
                           * & *                                                      & *                                                                                                      & \gamma I 
\end{bmatrix} \succ 0.
\end{align}
Moreover, note that since $P$ is positive and diagonal, for any matrix $M$ we have
\mmoh{$(P|M|)_{ij}=P_{ii}|M|_{ij}=|P_{ii}M_{ij}|=(|PM|)_{ij} \Rightarrow P|M|=|PM|$. This, combined with \eqref{eq:stability_1} and defining $\tilde{X} \triangleq PX, \forall X \in \{L,T,N\}$ yields the results}.
\end{proof}
  \begin{rem}{\mmoh{Note that the mixed-integer nature of the program in \eqref{eq:stability} stems from
  the existence of the absolute values. The main advantage of taking an MISDP-based approach is that it does not require us to impose further constraints, as is the case for the SDP-based solutions. In other word, the MISDP solution is \emph{tight}. On the other hand, while the MISDP in \eqref{eq:stability} involves more computational complexity compared to a typical SDP, it remains tractable using off-the-shelf solvers such as CUTSDP. Additionally, it is important to note that \eqref{eq:MISDP}
 is solved in an offline manner, i.e., it only needs to be computed once for any given design. So, its computational demands are quite minimal in contrast to online design methods. Alternatively, via imposing supplementary linear constraints, the error dynamics can be upper bounded by another linear comparison system with no absolute values involved. Hence, \eqref{eq:MISDP} can be relaxed to an SDP with reduced computational complexity. However, this comes at the cost of introducing more conservatism and sacrificing optimality.}}
 \end{rem}
 
\subsection{Guaranteed Privacy-Preserving
  Mechanism\label{sec:private-mech}}
  \mmoh{In this section, we provide additional conditions to ensure the proposed observer \eqref{eq:observer} preserves the privacy of output in a guaranteed manner (cf. Definition \ref{defn:guar_privacy}). To formalize our results, we} introduce a set-valued deterministic mechanism
(i.e., mapping $\mathcal{M}$) that \mmoh{is needed to be weakly/strongly} guaranteed privacy-preserving in the sense of Definition
\ref{defn:guar_privacy}. Note that given the measurement signal
$y_k$ and by Proposition \ref{prop:framer}, the to-be-designed
mechanism $\mathcal{M}$ outputs set-valued estimates of the desired variable
$z_k=\Gamma x_k \in [\underline{z}_k,\overline{z}_k]$: 
\begin{align}\label{eq:mechanism}
\begin{array}{c}
\mathcal{M}(y_k)=\mathcal{Z}_k \triangleq [\underline{z}_k,\overline{z}_k], \text{where}\\
\underline{z}_k,\overline{z}_k \ \text{are outputted by the observer \eqref{eq:observer}},
\end{array}
\end{align}  
\mmoh{where $x_k$ is the state of the perturbed system $\mathcal{G}_p$.} The following provides sufficient conditions for $\mathcal{M}$ to be a weakly guaranteed privacy-preserving mechanism.
\begin{thm}[\mmoh{Weakly} Guaranteed Privacy-Preserving Mechanism]\label{thm:privacy}
  Consider the perturbed system $\mathcal{G}_p$ in \eqref{eq:system_perturbed}. Let
  $\rho, \epsilon > 0$, 
  $(\gamma,\mmoh{P,{\widetilde{L}},\widetilde{T},\widetilde{N}})$ be a solution to the MISDP in
  \eqref{eq:stability}, and \mmoh{$L=P^{-1}{\widetilde{L},T=P^{-1}\widetilde{T},N=P^{-1}\widetilde{N}}$}. Then, the
  mechanism $\mathcal{M}$ defined in \eqref{eq:mechanism} is \mmoh{weakly}
  $\epsilon$-guaranteed private w.r.t. $\moo{\adj_{\rho}}$ given in
  \eqref{eq:adj},
  if 
  for some $\eta,\alpha >0$, the following matrix inequalities hold:
  \mmoh{
\begin{align}\label{eq:privacy}
  &\begin{bmatrix} P & \widetilde{T}A-{\widetilde{L}}C & [\widetilde{L} \ \widetilde{N}] & 0 \\
     * & P                                                   &  0                                                                                                     &  I \\
     * & *                                                     &  \eta I                                                                                         & 0 \\
     * & *                                                      & *                                                                                                      & P 
   \end{bmatrix} \hspace{-.1cm}\succ 0,\\
  \label{eq:new_inequality} &  \sigma_{m}(|\Gamma|)\gamma \alpha\|\delta_{\lambda}\|_2+2\sigma_m(\Gamma)\eta\rho \leq e^{\epsilon},
 \end{align}}
 where $\delta_{\lambda}\triangleq [\delta^\top_w \ \delta^\top_{v}]^\top$.
\end{thm}
\begin{proof}
  Consider two arbitrary
  $z_k \in
  \mathcal{M}(y_k)=\mathcal{Z}_k=[\underline{z}_k,\overline{z}_k]$ and
  $z_k' \in
  \mathcal{M}(y_k')=\mathcal{Z}'_k=[\underline{z}'_k,\overline{z}'_k]$,
  where $y_k,y_k' \in \mathcal{Y}$ such
  that 
  $\adj_\rho (y_k,y_k')$,
  i.e., 
  $\|y_k-y'_k\|_2 \leq \rho$.
  Defining
  $e^z_k \triangleq \overline{z}_k-\underline{z}_k, e^{z'}_k
  \triangleq \overline{z}'_k-\underline{z}'_k$ and using the triangle
  inequality we obtain:
  \begin{align}\label{eq:deltaz}
    \begin{array}{rl}
      \|z_k-z'_k\|_2&\hspace{-.2cm}=\|(z_k\hspace{-.1cm}-c^z_k)+(c^{z'}_k-z'_k)+(c^{z'}_k-c^{z}_k)\|_2\\
                    & \hspace{-.2cm}\leq \frac{1}{2}(\|e^z_k\|_2+\|e^{z'}_k\|_2)+\|\delta^{c_z}_k\|_2,
    \end{array}
  \end{align}
  where $\delta^{c_z}_k \triangleq c^z_k -c^{z'}_k$. Moreover, \mmoh{let}
  $c^{\nu}_k \triangleq \frac{1}{2}
  (\underline{\nu}_k+\overline{\nu}_k), \forall \nu \in
  \{z,z',x,x'\}$, \mmoh{be} the centers of the intervals
  $\mathcal{Z}_k,\mathcal{Z}'_k, \mathcal{X}_k\triangleq
  [\underline{x}_k,\overline{x}_k],\mathcal{X}'_k\triangleq
  [\underline{x}'_k,\overline{x}'_k]$, where
  $\mathcal{X}_k,\mathcal{X}'_k$ are the interval-valued state
  estimates returned by the proposed observer \eqref{eq:observer}
  using $y_k,y'_k$, respectively. Further, from \mmoh{\eqref{eq:observer} and \eqref{eq:error}}
  we have
  \begin{align}\label{eq:erro_eq}
    e^z_k=e^{z'}_k=|\Gamma|e^x_k,
  \end{align}
  where 
  the first equality follows from the facts that the error dynamics \eqref{eq:error} is not affected by the
  output signal $y_k$ and \mkj{that} the initial error
  $e^x_0=e^{x'}_0=\overline{x}_0-\underline{x}_0$ is a public
  information. 
    Combining
  \eqref{eq:deltaz} and \eqref{eq:erro_eq} returns
  \begin{align}\label{eq:deltaz2}
    \|z_k-z'_k\|_2 \leq \sigma_{m}(|\Gamma|) \|e^x_k\|_2+\|\delta^{c_z}_k\|_2.
  \end{align}
    Furthermore, from \eqref{eq:mechanism}, the definition of
  the interval center, and 
  $\Gamma^{\oplus}-\Gamma^{\ominus}=\Gamma$, we \mmoh{obtain}
  $c^z_k =\Gamma c^x_k, c^{z'}_k =\Gamma c^{x'}_k$. Hence,
 \begin{align}\label{eq:norm_e_xzt}
   \|\delta^{c_z}_k\|_2 \hspace{-.1cm}\triangleq \hspace{-.1cm} \|c^z_k\hspace{-.1cm} -\hspace{-.1cm}c^{z'}_k\|_2\hspace{-.1cm}=\hspace{-.1cm}\|\Gamma (c^x_k\hspace{-.1cm}-\hspace{-.1cm}c^{x'}_k)\|_2 \le
   \sigma_{m}(\Gamma) \|\delta^{c_x}_k\|_2,
 \end{align}
 where $\delta^{c_x}_k \triangleq c^x_k-c^{x'}_k$.
 This in addition to \eqref{eq:deltaz2} yields
  \begin{align}\label{eq:norm_e_xz}
\|z_k-z'_k\|_2  \leq \sigma_{m}(|\Gamma|) \|e^x_k\|_2+\sigma_{m}(\Gamma) \|\delta^{c_x}_k\|_2.
 \end{align}
 \mmoh{On the hand, it follows from} \eqref{eq:error_norm} and Lemma \ref{lem:stability} that:
 \begin{align}\label{eq:interm1}
 \|e^x_k\|_2 \leq \gamma \|\delta_{\lambda}\|_2. 
 \end{align}
 \mmoh{Combining  \eqref{eq:norm_e_xz} and \eqref{eq:interm1} returns:}
   \begin{align}\label{eq:norm_e_xz_2}
 \|z_k-z'_k\|_2  \leq \sigma_{m}(|\Gamma|) \gamma \|\delta_{\lambda}\|_2+\sigma_{m}(\Gamma) \|\delta^{c_x}_k\|_2.
 \end{align}
  \mmoh{Furthermore, from the first two equations in} \eqref{eq:observer}, we \mmoh{derive: 
  \begin{align}\label{eq:help1}
  \begin{array}{rl}
  \delta^{c_{\xi}}_{k+1}\hspace{-.2cm}&=\frac{1}{2}(\overline{\xi}_{k+1}+\underline{\xi}_{k+1}-(\overline{\xi}'_{k+1}+\underline{\xi}'_{k+1}))\\
  \hspace{-.2cm}&=\frac{1}{2}(\overline{\xi}_{k+1}-\overline{\xi}'_{k+1}+(\underline{\xi}_{k+1}-\underline{\xi}'_{k+1}))\\
  \hspace{-.2cm}&=\frac{1}{2}(TA\hspace{-.1cm}-\hspace{-.1cm}LC)(\overline{\xi}_{k}\hspace{-.1cm}+\hspace{-.1cm}\underline{\xi}_{k}-(\overline{\xi}'_{k}-\underline{\xi}'_{k}))+\tilde{y}_k-\tilde{y}'_k\\
  \hspace{-.2cm}&=(TA\hspace{-.1cm}-\hspace{-.1cm}LC)\delta^{c_{\xi}}_{k}\hspace{-.1cm}+\hspace{-.1cm}(L\hspace{-.1cm}+\hspace{-.1cm}(TA\hspace{-.1cm}-\hspace{-.1cm}LC)N)(y_k\hspace{-.1cm}-\hspace{-.1cm}y'_k).
  \end{array}
  \end{align} 
  Similarly, the two equations in \eqref{eq:observer} result in
  \begin{align}\label{eq:help2}
   \delta^{c_{x}}_{k}\hspace{-.1cm}= \delta^{c_{\xi}}_{k}+N(y_k-y'_k) \Leftrightarrow N(y_k-y'_k)= \delta^{c_{x}}_{k}-\delta^{c_{\xi}}_{k}.
  \end{align}
  Plugging the term $N(y_k-y'_k)$ from \eqref{eq:help2} into \eqref{eq:help1} yields:
 \begin{align}\label{eq:interm2}
 \delta^{c_x}_{k+1}\hspace{-.1cm}=\hspace{-.1cm}(A-LC)\delta^{c_x}_{k}\hspace{-.1cm}+\hspace{-.1cm}L(y_k\hspace{-.1cm}-\hspace{-.1cm}y'_k)\hspace{-.1cm}+\hspace{-.1cm}N(y_{k+1}-y'_{k+1}).
 \end{align}
 }
 By \mmoh{\cite[Lemma 2]{de2002extended}}, \eqref{eq:privacy} implies that
 the system in \eqref{eq:interm2} is stable and:
 \mmoh{
 \begin{align}\label{eq:delta_c_x_bounding}
 \|\delta^{c_x}_{k}\|_2 \leq \eta \| [(y_k-y'_k)^\top  (y_{k+1}-y'_{k+1})^\top]^\top \|_2 \leq 2\eta \rho,
 \end{align}
 }
\noindent where the factor $2$ in the right hand side comes from considering $y_k-y'_k$ and $y_{k+1}-y'_{k+1}$ as bounded-norm inputs to the system in \eqref{eq:delta_c_x_bounding}, which combined with \eqref{eq:new_inequality} and \eqref{eq:norm_e_xz_2} returns $\|z_k-z'_k\|_2  \leq \sigma_{m}(|\Gamma|) \gamma \|\delta_{\lambda}\|_2+2\sigma_{m}(\Gamma) \eta \rho \leq e^{\epsilon}$. 
\end{proof}
\mmoh{Theorem \ref{thm:privacy} provides sufficient conditions for weak privacy of the proposed interval observer, by guaranteeing that the interval-valued estimates of $z_k$ corresponding to two adjacent output values $y_k$ and $y'_k$, i.e., $\mathcal{Z}_k$ and $\mathcal{Z}'_k$, are close enough. Next we provide sufficient conditions for strong privacy to ensure that $\mathcal{Z}_k$ and $\mathcal{Z}'_k$ have a large enough intersection, which combined with the conditions in Theorem \ref{thm:privacy} implies that the interval observer design is strongly private.} 
 \mmoh{
 \begin{thm}[Strong Guaranteed Privacy]\label{thm:strong_privacy}
 Consider the perturbed system $\mathcal{G}_p$ and let $\Gamma$ has $n_z$ rows, and $\rho, \epsilon,\delta >0 $. Suppose the conditions in Theorems \ref{lem:stability} and 
\ref{thm:privacy} hold, and consequently, $(\gamma,\mmoh{P,{\widetilde{L}},\widetilde{T},\widetilde{N}})$ is a solution to the MISDP in \eqref{eq:stability}, \eqref{eq:privacy} and \eqref{eq:new_inequality}. Let $L=P^{-1}{\widetilde{L},T=P^{-1}\widetilde{T},N=P^{-1}\widetilde{N}}$. Then, there exists a time step $K \in \mathbb{N}$ such that for $k \geq K$ the mechanism $\mathcal{M}$ defined in \eqref{eq:mechanism} is \mmoh{strongly}
  $(\epsilon,\delta)$-guaranteed private w.r.t. $\moo{\adj_{\rho}}$ given in
  \eqref{eq:adj}, if 
  \begin{align}\label{eq:intersect}
  \alpha\||\Gamma|\Omega\delta_{\lambda}\|_{\infty} \geq \frac{1}{\sqrt{n_{z}}}\min (2 \sigma_{m}(\Gamma)\eta\rho,e^{\epsilon})+\delta,
  \end{align}
  where $\Omega \triangleq (I-(TA-LC))^{-1}[|TW| \quad |LV|+|NV|]$.
 \end{thm}
 \begin{proof}
 Since the mechanism is already weakly $\epsilon$-guaranteed private by Theorem \ref{thm:privacy}, all we need to prove is that \eqref{eq:intersect} implies:
 \begin{align*}
  \ttdiam(\mathcal{Z}_k \cap \mathcal{Z}'_k) \geq \delta \quad \forall k \geq K,
  \end{align*}
  which is equivalent to: $\forall k \geq K$,
  \begin{align}\label{eq:inters1}
  \exists i \in \{1,\dots,n_z \}, \ s.t. \ \frac{1}{2}(e^z_{k,i}+e^{z'}_{k,i})=e^z_i \hspace{-.1cm}\geq\hspace{-.1cm} |\delta^{c_z}_{k,i}|\hspace{-.1cm}+\hspace{-.1cm}\delta,
  \end{align}
  where the equality holds since $e^z=e^{z'}$. Given $|e^z_{k,i}|=e^z_{k,i}$, a sufficient condition for \eqref{eq:inters1} is $\forall k \geq K$,
   \begin{align*}
 \exists i \in \{1,\dots,n_z \}, \ s.t. \ \max_{i}|e^z_{k,i}|=\|e^z_k\|_{\infty} \geq |\delta^{c_z}_{k,i}|+\delta,
  \end{align*}  
  or, equivalently: $\forall k \geq K$,
   \begin{align}\label{eq:inters2}
  \|e^z_k\|_{\infty} \geq \max_{i}|\delta^{c_z}_{k,i}|+\delta=\|\delta^{c_z}_k\|_{\infty}+\delta.
  \end{align}  
  On the other hand, it follows from \eqref{eq:norm_e_xzt} and \eqref{eq:delta_c_x_bounding} that:  
 $\|\delta^{c_z}_k\|_2 \leq  \sigma_{m}(\Gamma) \|\delta^{c_x}_k\|_2 \leq 2\sigma_{m}(\Gamma) \eta \rho$, 
  while by weak privacy: $\|\delta^{c_z}_k\|_2 \leq e^\epsilon$. Hence, $\|\delta^{c_z}_k\|_2 \leq \min (2\sigma_{m}(\Gamma) \eta \rho,e^\epsilon)$ implying $\|\delta^{c_z}_k\|_{\infty} \leq \frac{1}{\sqrt{n_z}} \min (2\sigma_{m}(\Gamma) \eta \rho,e^\epsilon)$. So, for \eqref{eq:inters2} to hold, it suffices that $\forall k \geq K$,
   \begin{align}\label{eq:inters3}
  \|e^z_k\|_{\infty}=\||\Gamma|e^x_k\|_{\infty} \geq \frac{1}{\sqrt{n_z}} \min (2\sigma_{m}(\Gamma) \eta \rho,e^\epsilon)+\delta.
  \end{align} 
  If we can show that the above inequality holds for the steady state error, i.e., $e^x_{ss}$, then the fact that $\lim_{k \to \infty}e^x_{k}=e^x_{ss}$ together with the continuity of the $\lim$ operator implies \eqref{eq:inters3} for $k$ greater than some sufficiently large $K$. To do this, solving \eqref{eq:error} using simple induction returns:
  \begin{align*}
\begin{array}{rl}
  e^{x}_{k}
  &=\widetilde{A}^ke^x_0+(\sum_{i=1}^k\widetilde{A}^{k-i})\Lambda
    \delta_{\hat{\lambda}}\\
    &=\hspace{-.1cm}\widetilde{A}^ke^x_0\hspace{-.1cm}+\hspace{-.1cm}(I\hspace{-.1cm}-\hspace{-.1cm}\widetilde{A})^{-1}(I-\widetilde{A}^k)\Lambda
    \delta_{\hat{\lambda}} \xrightarrow{k \to \infty} \overbrace{(I-\widetilde{A})^{-1}\Lambda
    \delta_{\hat{\lambda}}}^{e^x_{ss}}.
\end{array}
\end{align*}
where $\widetilde{A} \triangleq |TA-LC|$, $\Lambda$ and $\delta_{\hat{\lambda}}$ are given in \eqref{eq:error}, and the limit holds since the comparison system \eqref{eq:error_comparison} is stable by Theorem \ref{lem:stability}. Plugging $e^k_{ss}$ into \eqref{eq:inters3} returns the results. 
 \end{proof}
 }
\subsection{$\mathcal{H}_{\infty}$-Optimal Guaranteed Private Observer
  Synthesis}\label{sec:private_observer}
Finally, combining the results of \mmoh{Theorems~\ref{lem:stability}--\ref{thm:strong_privacy}} enables us to find \mmoh{observer} gains and perturbation factors \mmoh{$\alpha$ and $\delta$} that
simultaneously stabilize the observer, \mkj{satisfy} the guaranteed privacy
and minimize the $\mathcal{H}_{\infty}$-norm of the error
system. \mmoh{To do that, for given $\rho$ and $\epsilon$, we first design a weakly $\epsilon$-guaranteed mechanism $\mathcal{M}$ \mkj{w.r.t.}
$\moo{\adj_{\rho}}$, using the results in Theorems~\ref{lem:stability} and~\ref{thm:privacy}. Then, leveraging \eqref{eq:intersect} and Theorem \ref{thm:strong_privacy}, we find the largest $\delta$ with which $\mathcal{M}$ is strongly $(\epsilon,\delta)$-guaranteed private \mkj{w.r.t.}
$\moo{\adj_{\rho}}$, if possible. Otherwise, we reduce the adjacency distance $\rho$ to $\tilde\rho \in (0,\rho)$ until strong $(\epsilon,\tilde\delta)$- privacy is obtained w.r.t.~$\moo{\adj_{\tilde\rho}}$ with some $\tilde{\delta} >0$. The following theorem summarizes this process showing that it is always feasible to find such $\tilde\rho$ and $\tilde{\delta}$.}
\mmoh{
\begin{thm} [$\mathcal{H}_{\infty}$-Optimal Guaranteed Private
  Observer Synthesis]\label{thm:synthesis}
  Consider system \mkj{$\mathcal{G}$} and the corresponding perturbed system $\mathcal{G}_p$.
  Let
  $\rho, \epsilon >0$ and suppose
  $(\gamma_*,\beta_*,\eta_*,\mmoh{P_*},{\widetilde{L}}_*,{\widetilde{T}}_*,{\widetilde{N}}_*)$ be a solution to the following: 
\begin{align}\label{eq:MISDP}
  &\min\limits_{\{\gamma>0,\beta>0,\eta>0,P\in \mathbb{D}^{\moo{n}\times\moo{n}}_{>0},{\widetilde{L}},{\widetilde{N}},{\widetilde{T}}\}}\gamma\\
  \nonumber&\quad \quad \quad \quad   \st \eqref{eq:stability}, \eqref{eq:privacy}, \ \sigma_{m}(|\Gamma|)\beta\|\delta_{\lambda}\|_2+2\sigma_m(\Gamma)\eta\rho \leq e^{\epsilon}.
\end{align}
Then, the proposed observer \eqref{eq:observer} with the gains
$L_*=P_*^{-1}\moo{\widetilde{L}}_*,T_*=P_*^{-1}\moo{\widetilde{T}}_*,N_*=P_*^{-1}\moo{\widetilde{N}}_*$ and the perturbation factor $\alpha_*=\frac{\beta_*}{\gamma_*}$ is
$\mathcal{H}_{\infty}$-optimal in the sense of Definition
\ref{defn:opt_obs}, with the optimal noise attenuation level
$\gamma^*$. Moreover, the mechanism $\mathcal{M}$ defined in
\eqref{eq:mechanism} is weakly $\epsilon$-guaranteed private \mkj{w.r.t.}
$\moo{\adj_{\rho}}$ given in \eqref{eq:adj}. Furthermore, the following hold.
\begin{enumerate}[(i)]
\item\label{item:strong_1} If $\alpha_*\||\Gamma|\Omega_*\delta_{\lambda}\|_{\infty} > \frac{1}{\sqrt{n_{z}}}\min (2 \sigma_{m}(\Gamma)\eta_*\rho,e^{\epsilon})$, then there exists a time step $K \in \mathbb{N}$ such that for $k \geq K$, $\mathcal{M}$ is strongly $(\epsilon, \delta)$-guaranteed private \mkj{w.r.t.} $\moo{\adj_{\rho}}$ with $\delta=\alpha_*\||\Gamma|\Omega_*\delta_{\lambda}\|_{\infty} - \frac{1}{\sqrt{n_{z}}}\min (2 \sigma_{m}(\Gamma)\eta_*\rho,e^{\epsilon})$, where $\Omega_* \triangleq (I-(T_*A-L_*C))^{-1}[|T_*W| \quad |L_*V|+|N_*V|]$.
\item\label{item:strong_2} If $\alpha_*\||\Gamma|\Omega_*\delta_{\lambda}\|_{\infty} \leq \frac{1}{\sqrt{n_{z}}}\min (2 \sigma_{m}(\Gamma)\eta_*\rho,e^{\epsilon})$, then there exists a time step $\tilde{K} \in \mathbb{N}$ such that for $k \geq \tilde{K}$, 
$\mathcal{M}$ is strongly $(\epsilon, \tilde\delta)$-guaranteed private \mkj{w.r.t.} $\moo{\adj_{\tilde \rho}}$ with any $\tilde{\rho} \in (0,\frac{\sqrt{n_z}\alpha_*\||\Gamma|\Omega_*\delta_{\lambda}\|_{\infty}}{2\sigma_m(\Gamma)\eta_*}]$ and $\tilde\delta=\alpha_*\||\Gamma|\Omega_*\delta_{\lambda}\|_{\infty} - \frac{2 \sigma_{m}(\Gamma)\eta_*\tilde\rho}{\sqrt{n_{z}}}$.
\end{enumerate}
\end{thm}
\begin{proof}
  The proof of weak privacy directly follows from Theorems~\ref{lem:stability} and~\ref{thm:privacy}, and defining the new decision variable
  $\beta=\gamma\alpha$. The statement in \eqref{item:strong_1} holds by \eqref{eq:intersect} and Theorem \ref{thm:strong_privacy}. Finally, the statement in \eqref{item:strong_2} follows from the facts that the right hand side of \eqref{eq:intersect} is decreasing in $\rho$. So, if \eqref{eq:intersect} does not hold with a given $\rho$, one can reduce $\rho$ as much as needed to $\tilde{\rho}$ until \eqref{eq:intersect} holds. In this case, it is easy to see the largest $\tilde{\rho}$ that satisfies \eqref{eq:intersect} is $\rho_m=\frac{\sqrt{n_z}\alpha_*\||\Gamma|\Omega_*\delta_{\lambda}\|_{\infty}}{2\sigma_m(\Gamma)\eta_*}$, i.e., \eqref{eq:intersect} holds for any $\tilde{\rho} \in (0,\rho_m]$ and the corresponding $\tilde{\delta}$.
\end{proof}
}
\section{Accuracy Analysis}
As a consequence of introducing perturbations to ensure guaranteed
privacy, the estimates of $x$ and $z$ incur into an accuracy loss.  In
this section, we 
quantify the difference between the interval estimate widths, i.e., the
observer errors, with and without considering guaranteed privacy.
 
First, it is straightforward to see that in the absence of privacy
considerations, a non-private (NP) $\mathcal{H}_{\infty}$-optimal
interval observer can be designed by implementing~\eqref{eq:observer}
without any additional perturbation noise, i.e, with $\alpha^{\NP}=1$,
and with observer gains $L^{\NP}=(P^{\NP}_*)^{-1}\moo{\widetilde{L}}^{\NP}_*,T^{\NP}=(P^{\NP}_*)^{-1}\moo{\widetilde{T}}^{\NP}_*,N^{\NP}=(P^{\NP}_*)^{-1}\moo{\widetilde{N}}^{\NP}_*$,
where $(P^{\NP}_*,\moo{\widetilde{L}}^{\NP}_*,\moo{\widetilde{T}}^{\NP}_*,\moo{\widetilde{N}}^{\NP}_*,\gamma^{\NP}_*)$ is a solution to
the following MISDP:
\begin{align}\label{eq:nonprivate}
\begin{array}{rl}
  &\min\limits_{\gamma >0, P \in \mathbb{D}^{\moo{n}\times\moo{n}}_{>0},\mmoh{\widetilde{L}},\mmoh{\widetilde{T}},\mmoh{\widetilde{N}}} \gamma\\
  &\quad \quad \ \st \, \text{\sm{matrix inequality} } \eqref{eq:stability} \text{ \sm{holds}}.
\end{array}
\end{align}
Then, the corresponding NP state framers are
$\underline{x}^{\NP},\overline{x}^{\NP}$.

On the other hand, in the presence of privacy,
  the
proposed observer \eqref{eq:observer} outputs the guaranteed-private
(GP) framers
$\underline{x}_k=\underline{x}^{\GP}_k,\overline{x}_k=\overline{x}^{\GP}_k$,
with $\alpha=\frac{\beta^{\GP}_*}{\gamma^{\GP}_*}$ and
$X^{\GP}=(P^{\GP}_*)^{-1}\moo{\widetilde{X}}^{\GP}_*, X \in \{L,T,N\}$, where
$(\eta^{\GP}_*,\beta^{\GP}_*,P^{\GP}_*,\moo{\widetilde{L}}^{\GP}_*,\moo{\widetilde{T}}^{\GP}_*,\moo{\widetilde{N}}^{\GP}_*,\gamma^{\GP}_*)$
is a solution to \eqref{eq:MISDP} (cf. Section
\ref{sec:private_observer}).

Consequently, for $s \in \{\NP,\GP\}$, the following intervals
\begin{align*}
  \begin{array}{c}
    \mathcal{Z}^{s}_k=[\underline{z}^{s}_k,\overline{z}^{s}_k],\\
    \underline{z}^{s}_k
    =\Gamma^{\oplus}\underline{x}^{s}_k-
    \Gamma^{\ominus}\overline{x}^{s}_k, \  \overline{z}^{s}_k=
    \Gamma^{\oplus}\overline{x}^{s}_k-\Gamma^{\ominus}\underline{x}^{s}_k
  \end{array}
\end{align*}
are the corresponding non-private (if $s=\NP$) and guaranteed-private
(if $s=\GP$) interval estimates of $z_k=\Gamma x_k$,
respectively. This being said, we are interested in quantifying an
upper bound for the accuracy error:
 \begin{align}\label{eq:acc_error}
   \varepsilon^p_k\triangleq \hspace{-.1cm} \|e^{z,\NP}_k-e^{z,\GP}_k\|_{\infty}\hspace{-.1cm}=\hspace{-.1cm}\|(\overline{z}^{\GP}_k-\underline{z}^{\GP}_k)\hspace{-.1cm}-\hspace{-.1cm}(\overline{z}^{\NP}_k\hspace{-.1cm}-\hspace{-.1cm}\underline{z}^{\NP}_k)\hspace{-.05cm}\|_{\infty}, 
 \end{align}
 which measures the conservatism incurred by considering guaranteed
 privacy. The following lemma formalizes this by providing a
 closed-from quantification of the accuracy error and its steady
 state value.
\begin{lem}\label{lem:accuracy}
  Consider \mkj{system $\mathcal{G}$} and the proposed guaranteed
  private $\mathcal{H}_{\infty}$-optimal observer
  \eqref{eq:observer}. Suppose all the assumptions in Theorem
  \ref{thm:synthesis} hold. Then, the accuracy error $\varepsilon^p_k$
  (cf. \eqref{eq:acc_error}) and its steady state value
  $\varepsilon^p_{\infty}$ can be computed as: 
\moo{
\begin{align*}
\begin{array}{rl}
 \varepsilon^p_k&=\| |\Gamma| (\widetilde{\Delta}_ke^x_0+\Theta_k{\Lambda}\delta_{{\lambda}}) \|_{\infty},\\
 \varepsilon^p_{\infty}&=\| |\Gamma| (({I-\subscr{\widetilde{A}}{NP}})^{-1}-\alpha({I-\subscr{\tilde{A}}{GP}})^{-1}){\Lambda}\delta_{{\lambda}} \|_{\infty},
 \end{array}
\end{align*}
where \sm{recall that} $e^x_0 \triangleq \overline{x}_0-\underline{x}_0$, $\widetilde{\Delta}_k\triangleq \subscr{\widetilde{A}}{NP}^k-\subscr{\widetilde{A}}{NP}^k$, and
\begin{align*}
\Theta_k \triangleq (({I-\subscr{\widetilde{A}}{NP}})^{-1}({I-\subscr{\widetilde{A}}{NP}^k})-\alpha({I-\subscr{\tilde{A}}{GP}})^{-1}({I-\subscr{\widetilde{A}}{GP}^k})),
\end{align*}
  \sm{with} $ \subscr{\widetilde{A}}{s} \triangleq |TA-\supscr{L}{s}C|$,
\sm{for} $ {s} \in \{\rm{NP},\rm{GP}\}$,
${\Lambda} \triangleq  [|TW| \quad |LV|+|NV|]$, and
$\delta_{{\lambda}} \triangleq [\delta_{w}^\top \
\delta_{v}^\top]^\top$.}
\end{lem}
\begin{proof}
  Defining $e^{x,s}_k \triangleq \overline{x}^s_k- \underline{x}^s_k,$
  \sm{for} $ s \in \{\NP,\GP\}$, and computing the error dynamics of
  the observer \eqref{eq:observer} for each of the two NP and GP cases
  yields:
  \moo{
\begin{align}\label{eq:NP_GP_error}
  e^{x,s}_{k+1}=\widetilde{A}_se^{x,s}_{k}+\Lambda_s\delta_{\lambda}, \quad s \in \{\NP,\GP\}. 
\end{align}}
Solving \eqref{eq:NP_GP_error} using simple induction returns:
\begin{align*}
\begin{array}{rl}
  e^{x,s}_{k}
  &=\widetilde{A}^k_se^x_0+(\sum_{i=1}^k\widetilde{A}_s^{k-i})\Lambda_s
    \delta_{{\lambda}}\\ 
  &=\widetilde{A}^k_se^x_0+({I-\widetilde{A}_s})^{-1}({I-\widetilde{A}^k_s})
   \Lambda_s
    \delta_{{\lambda}}, \quad s \in \{\NP,\GP\}. 
\end{array}
\end{align*}
The closed-form expression for the accuracy error follows from this
and the fact that
$e^{z,s}_k \triangleq |\Gamma| e^{x,s}_k, \forall s \in
\{\NP,\GP\}$. Finally, $\varepsilon^p_{\infty}\triangleq \lim_{k \to \infty}\varepsilon^p_k$ is computed by taking
the limit of $\varepsilon^p_k$ when time goes to infinity and given
that
\moo{$\lim_{k\to \infty} \widetilde{A}^k_s=\mathbf{0}_{\moo{n}\times\moo{n}},
\forall s \in \{\NP,\GP\}$} by the observer stability.  
\end{proof}
  \section{Illustrative Example}\label{sec:example}
  To illustrate the effectiveness of our approach, we consider a
  slightly modified version of a dynamic market, \mkj{i.e., a L\"{e}ontief model,} with $N=5$ firms that
  supply the same
  product~\cite{zhang2013adaptive,degue2020differentially},
    where a linear system models the production
  dynamics of each firm, which is affected by its neighbor firms: $  x^i_{k+1}=(1-a)x^i_k+w^i_k+\frac{a}{|\mathcal{N}_i|}\textstyle{\sum}_{j \in \mathcal{N}_i}x^j_k, \quad k
    \in \mathbb{K}.$ 
  Here, $x^i_k$ and $\mathcal{N}_i$ represent the production output
  of the firm~$i$ and the set of its neighbors, respectively. Moreover,
  $w^i_k \in [-0.5,0.5]$ is assumed to be an individual firm
  production disturbance. Furthermore, similar to
  \cite{zhang2013adaptive}, we set $a=0.16$,
  $\mathcal{N}_i=\{i+1\}, \forall i \in \{1,\dots,N-1\}$ and
  $\mathcal{N}_N=\{1\}$. {In other words, the firm \emph{influence graph} is a ring graph where all connected nodes are labeled by consecutive
    integers.}. 
    Each firm $i$, has an uncertain initial production
  $x^i_0 \in [185,215]$ and shares a noisy measurement of its highly
  sensitive production output $y^i_k=x^i_k+v^i_k$, where
  $v^i_k \in [0,1]$. \mkj{So}, we obtain global state and
  measurement equations in the form of \eqref{eq:system}, with
  $\overline{x}_0=215\mathbf{1}_5,\underline{x}_0=185\mathbf{1}_5,
  \overline{w}_0=-\underline{w}_0=0.5\mathbf{1}_5,\overline{v}_0=\mathbf{0}_5,\underline{v}_0=\mathbf{1}_5,V=W=I$, {$A_{ii}=1-a$, $ A_{i,i+1}=a$, and the rest of the element of $A$ are zero.}
  Furthermore, $\Gamma=\mathbf{1}^\top_5$, i.e., there is a data
  aggregator planning to release the total production output
  $z_k=\sum_{i=1}^5x^i_k$ at each time step (day) from the data of the
  $N=5$ firms, in a way that the \mmoh{strong} $(\epsilon,\delta)$-guaranteed
  privacy of the data is preserved \mkj{w.r.t.} $\moo{\adj_{\rho}}$,
  with $\epsilon=\ln(3),\delta=.1$ and $\rho=1$. This
    global production can be of value to a decision maker to make
    future investments. However, in releasing this, an operator would
    like to keep each firm's information private. By applying Theorem
  \ref{thm:synthesis} and solving the MISDP in \eqref{eq:MISDP} using
  the solver CUTSDP \mkj{implemented in} YALMIP \cite{Lofberg2004}, we
  obtain the \mmoh{corresponding noise perturbation factor
    $\alpha_*=1.273$, the noise attenuation level $\gamma_*=0.754$,
    the observer gains $(L_*)_{ii}=0.314,(T_*)_{ii}=1.205,(N_*)_{ii}=-0.205,(L_*)_{i,i+1}=0.053,(T_*)_{i,i+1}=-(N_*)_{i,i+1}=0.107$, and the
    rest of the elements of $L_*$ and $N_*$ all $-.004$ and $0.006$, respectively.}
The red plots in Figure \ref{fig:privacy_estimates}
indicate the guaranteed private (GP) upper and lower framers (left) and the estimate interval widths (right) \mkj{of
$z_k$}, obtained by observer \mkj{$\widetilde{\mathcal{G}}$} with the
computed $L_*$ and $\alpha_*$. As can be seen, the plotted framers
contain the actual state trajectory (the green plot). Moreover, as
expected, the non-private (NP) interval estimates (black plots) are
tighter than the GP ones due to the additional required guaranteed
privacy-preserving constraints \eqref{eq:privacy} and
\eqref{eq:new_inequality}, as well as the additional perturbation
noise. \mmoh{Furthermore}, for the sake of comparison, we implemented a slightly
modified version of the differentially private (DP) interval observer
in \cite{degue2020differentially}, using our computed gain \mmoh{${L}_*,T_*,N_*$}, by
perturbing the input data $y_k$ with additional stochastic noise as
described in \cite{degue2020differentially}. As Figure
\ref{fig:privacy_estimates} shows, the GP interval estimates (red)
are tighter the DP ones (blue). \mmoh{Finally, Figure \ref{fig:intersect} shows the interval-valued estimates corresponding to two adjacent output signals $y^1_k$ and $y^2_k$. As can be seen, all the intervals have non-empty intersection and are hard to be distinguished which illustrates the strong privacy-preserving feature of the designed mechanism.}
\begin{figure}[t!] 
\centering
{\includegraphics[width=\columnwidth,trim=5mm -5mm 10mm 0mm]{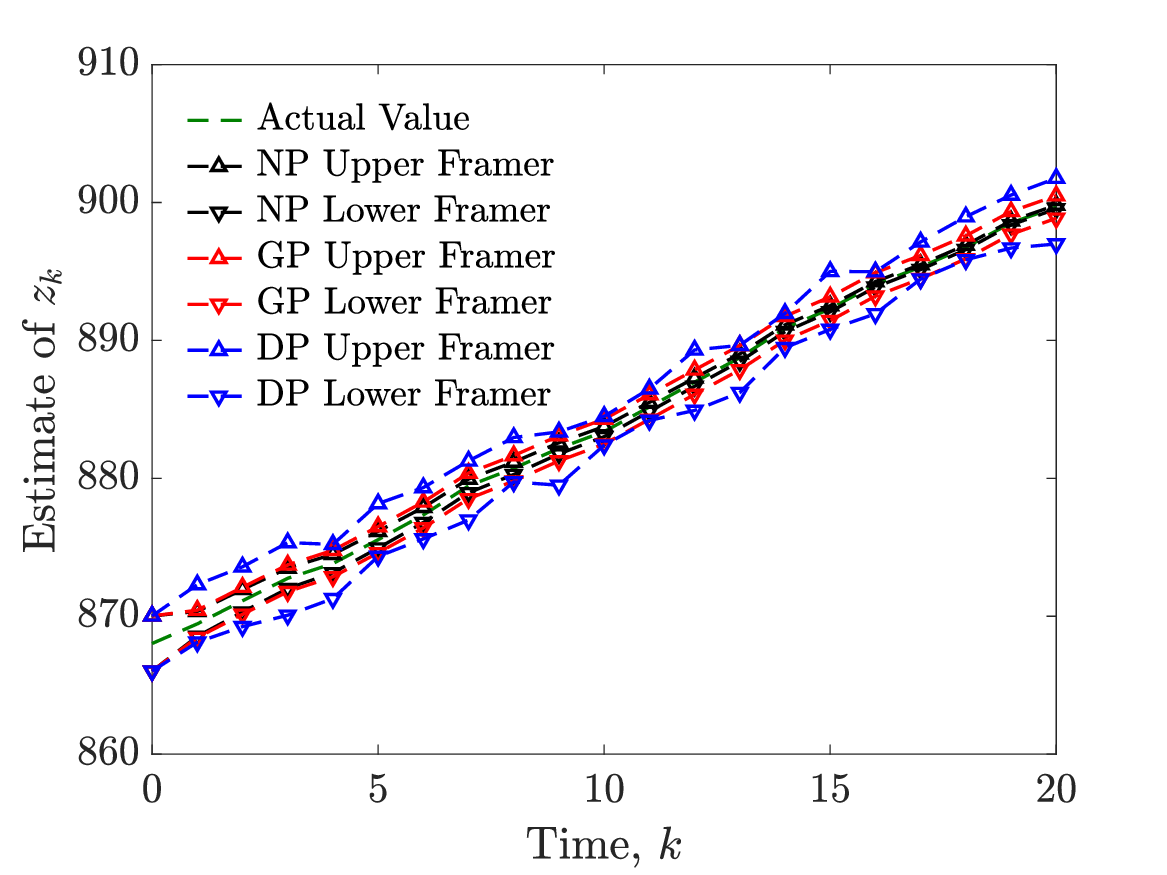}} 
{\includegraphics[width=\columnwidth,trim=5mm -5mm 10mm 0mm]{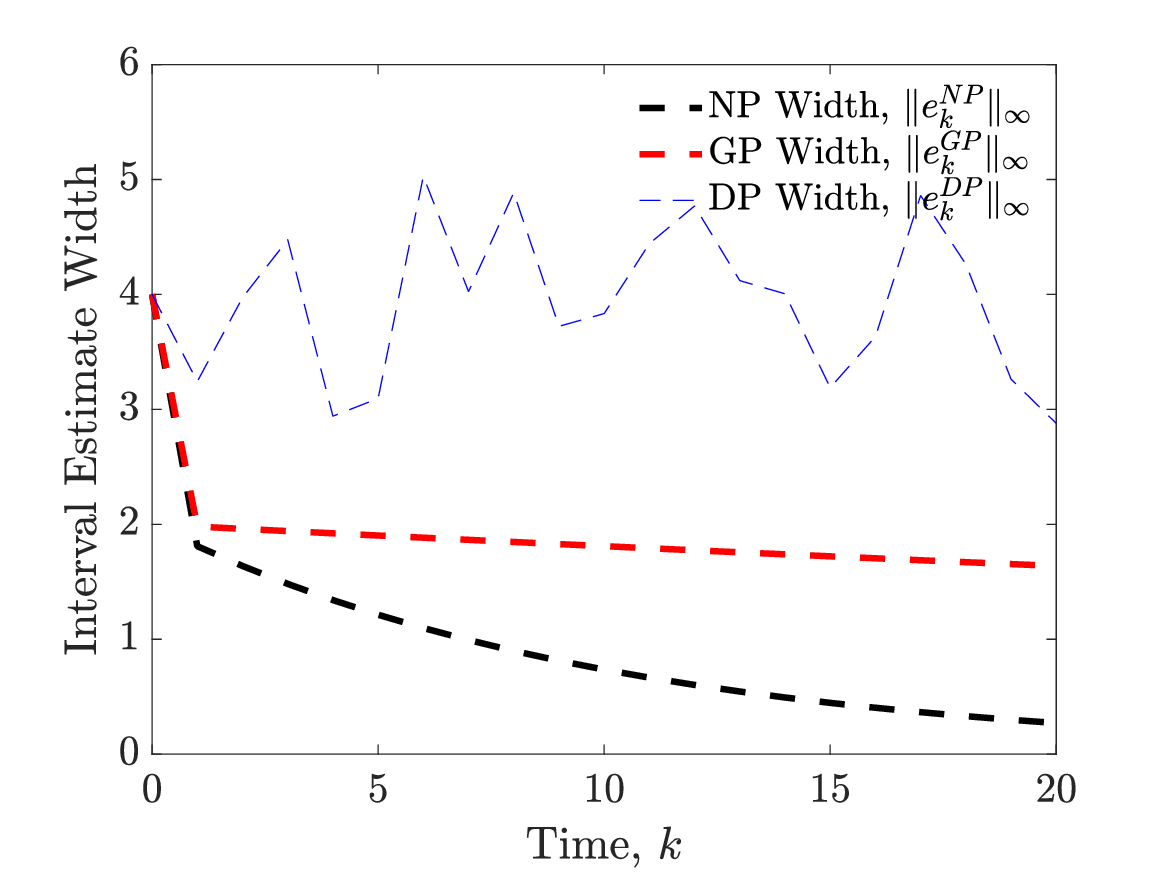}}
\vspace{-.25cm}
\caption{{{\small \textbf{Top}: actual value of $z_k$ (green), as
      well as its upper and lower framers obtained by applying a
      non-private (NP) interval observer (black), a guaranteed private
      (GP) interval observer (red), and a differentially private (DP)
      interval observer (blue). {\textbf{Bottom}: comparison of the
        no-private (NP), the guaranteed private (GP) and
        differentially private (DP)\cite{degue2020differentially}
        interval estimate widths (errors).}}}}
\label{fig:privacy_estimates}
\end{figure}
\begin{figure}
\includegraphics[width=\columnwidth,trim=5mm -5mm 10mm 0mm]{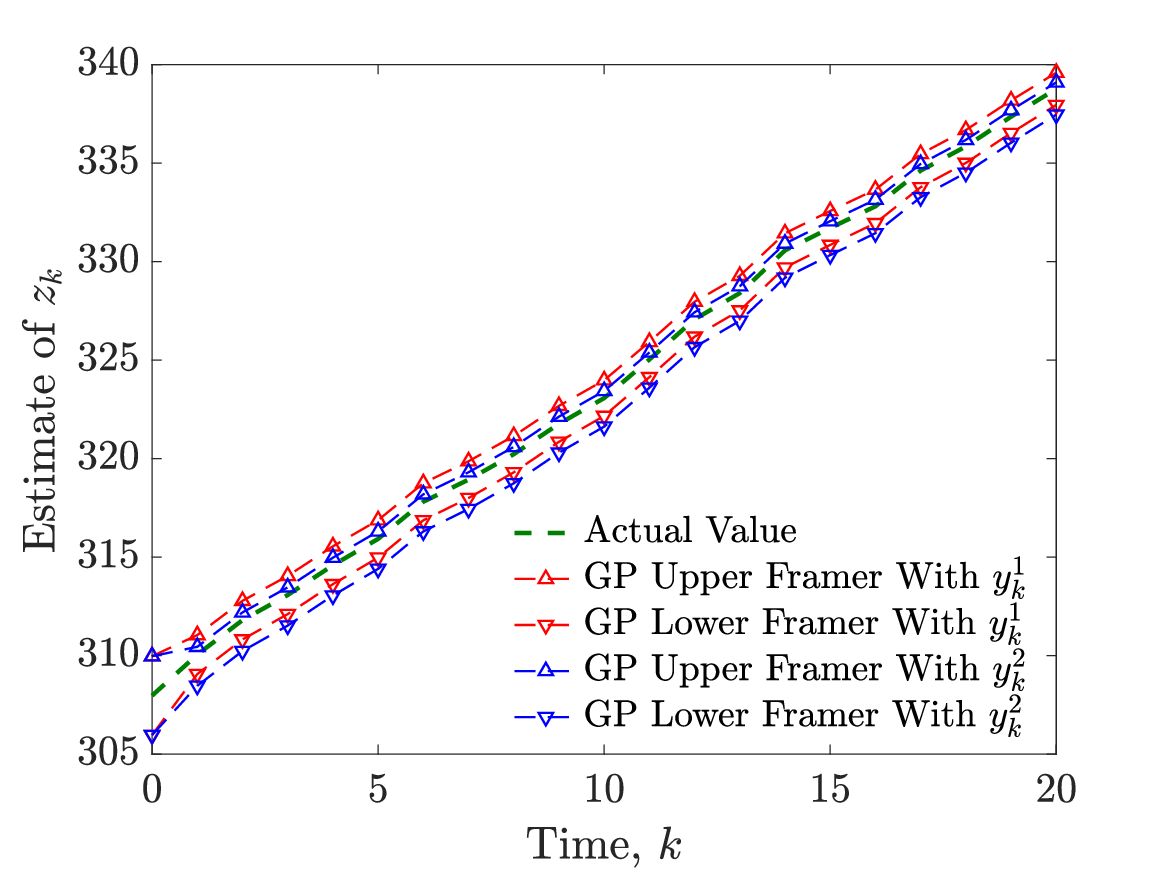}
\caption{{{\small Strong guaranteed privacy-preserving interval-valued estimates of $z_k$ using two adjacent outputs $y^1_k$ (red) and $y^2_k$ (blue).}}}
\label{fig:intersect}
\end{figure}
\section{Conclusion and \md{Future Work}} \label{sec:conclusion} A
novel generalization of guaranteed privacy was proposed in this paper,
affording a deterministic portrayal of privacy. In contrast to
stochastic differential privacy, guaranteed privacy was found to
impose deterministic constraints on the proximity between the ranges of two
sets of estimated data and their intersection. \moo{To do so}, 
  an interval observer was designed \mkj{for a perturbed bounded-error LTI system}, incorporating \mkj{a} bounded
noise perturbation factor and an observer gain. \mkj{The} observer simultaneously \mkj{outputted} guaranteed private and
stable interval-valued estimates for the desired variable. The
optimality of the design was demonstrated, and the accuracy of the mechanism was assessed by quantifying the loss incurred when
considering guaranteed privacy specifications. Future work will
consider nonlinear systems and combination of privacy and attack
resilience.  \bibliographystyle{unsrturl}

{\tiny
\bibliography{biblio}
}

\end{document}